\documentclass[10pt]{article}

\usepackage{amssymb}
\usepackage{amsfonts}
\usepackage{amsmath}
\usepackage{amscd}
\usepackage[all,cmtip]{xy}
\usepackage{amsthm}
\usepackage{setspace}
\usepackage{enumerate}
\usepackage{geometry}
\usepackage{tcolorbox}
\usepackage{tikz}
\usepackage{hyperref}

\hypersetup{
allbordercolors=blue, pdfborder=0 0 1.3}

\theoremstyle{plain}
\newtheorem{thm}{Theorem}

\newtheorem*{cor}{Corollary}

\theoremstyle{definition}
\newtheorem*{remark}{Remark}
\newtheorem{condition}{Condition}
\newtheorem{example}{Example}

\newcommand{\ZZ}{{\mathbb Z}}

\title{The KdV vacuum operator and its Katz extension}
\author{Martin T. Luu}

\date{}

\begin{document}

\maketitle

\begin{abstract}
We define a connection on the formal disc that can be used to single out the vacuum of the Drinfeld-Sokolov KdV hierarchy associated to a simple complex finite-dimensional Lie algebra. As a connection, it has a canonical Katz extension from the disc to the sphere.  We express this Katz extension in terms of the Kac coordinates of a suitable Weyl group conjugacy class. As a consequence, we show that the Katz extension has meaning in the context of the integrable hierarchy: It describes an additional symmetry.
\end{abstract}

\section{Introduction}

Let $\mathfrak g$ be a finite-dimensional simple complex Lie algebra.  Kostant initiated in \cite{KOS} the study of cyclic elements $\Lambda$ associated to $\mathfrak g$.  To define them,  let $\{e_{i},f_{i},h_{i}\}_{1\le i \le r}$ be generators of $\mathfrak g$ satisfying the Serre relations. Let $F$ be a non-zero element in the lowest root space with respect to the Cartan algebra $\widetilde{\mathfrak h}$ generated by the $h_{i}$'s. A (classical) cyclic element in $\mathfrak g$ is an element of the form $\sum_{i=1}^{r} e_{i} + F$. 
Analogously, one defines the affine variant to be an element $\Lambda$ in the loop algebra $\widehat{\mathfrak g}=\mathfrak g[z,z^{-1}]$ of the form
$$\Lambda=  \sum_{i=1}^{r} e_{i} + z \cdot F$$
After choosing a complex $d$-dimensional representation $\xi$ of $\mathfrak g$,  the cyclic element $\Lambda$ acts on 
$$|0\rangle :=\mathbb{C}[z]^{d}$$
The notation stems from the relation to Drinfeld-Sokolov integrable hierarches, where $|0\rangle$ describes the vacuum in phase space.  Similarly to the $\Lambda$ action,  the vacuum $|0\rangle$ is stabilized by the linear differential operator 
$$\nabla=\partial_{z}+ \Lambda$$
Operators closely related to $\nabla$ have been studied in various contexts recently: See the work of Frenkel-Gross \cite{FG} related to Langlands correspondences, the work of Masoero-Raimondo-Valeri \cite{MRV} on the ODE/IM correspondence,  as well as the work of Kamgarpour-Sage \cite{KAMS}.  Note that a variant $\partial_{x}+\Lambda$ with $x$ an indeterminate different than the loop variable $z$ is a central object in the Lax operator formulation of the Drinfeld-Sokolov integrable hierarchy associated to $\mathfrak g$.

Our first aim is to show a special relation between $\nabla$ and the vacuum $|0\rangle$.  To make this precise,  consider deformations of $\mathbb{C}[z]^{d}$ inside of $\mathbb{C}(\!(1/z)\!)^{d}$ that are complex subspaces of $\mathbb{C}(\!(1/z)\!)^{d}$ whose projection onto $\mathbb{C}[z]^{d}$ is an isomorphism. Among these, we consider deformations coming from $\mathfrak g$ in the following sense: Let $U$ in $\widehat{\mathfrak g}$ be of the form $\sum_{i<0} U_{i}$ where each $U_{i}$ is of degree $i$ in the principal gradation and also in $z^{-1} \mathfrak g[z^{-1}]$. Then by a $(\mathfrak g, \xi)$-deformation of $|0\rangle $ we mean a subspace of $\mathbb{C}(\!(1/z)\!)^{d}$ of the form $V=\textrm{exp}(U) |0\rangle$ where $U$ is as above and acts via the representation $\xi$. A natural question: Is the vacuum the only point in this collection of spaces that is stabilized by $\partial_{z}+\Lambda$? We show this is indeed the case:  

\begin{thm}
\label{main-theorem}
If $V$ is a $(\mathfrak g, \xi)$-deformation of $|0\rangle $ such that $\nabla \; V \subseteq V$,  then $V=|0\rangle$. 
\end{thm}

Motivated by this result,  we will call $\nabla$ the vacuum operator.  Theorem \ref{main-theorem} is a purely Lie theoretic result, but the proof we give heavily relies on ideas coming from integrable systems, in particular certain uniqueness results related to the Witten-Kontsevich points. The role of integrable systems should not be surprising,  Cafasso and Wu show in \cite{CAW1} (Theorem 3.5) that one realization of the phase space of the Drinfeld-Sokolov hierarchy associated to $\mathfrak g$ and $\Lambda$ is precisely the space of $(\mathfrak g, \xi)$-deformations of the vacuum (this kind of Grassmannian description has a long history, certainly going back to the description of the KP hierarchy).

Special cases of Theorem \ref{main-theorem} can be completely elementary: Consider $A_{2}$ in its standard representation $\mathfrak g= \mathfrak s \mathfrak l_{3}$ (we will see later on that if the theorem is proved for one choice of faithful representation $\xi$ of $\mathfrak g$ then the theorem is proved for all such representations). One can choose $\Lambda$ as 
$$\Lambda=e_{2,1}+e_{3,2} + z\cdot e_{1,3}$$ 
where $e_{i,j}$ has $0$'s everywhere except a $1$ at the $(i,j)$ entry.
Note that for each $(\mathfrak g,\xi)$-deformation $V$ of $|0\rangle$,  the natural projection
$$\pi : V \rightarrow \mathbb{C}[z]^{3}=|0\rangle $$
is an isomorphism and $zV\subseteq V$. Let $v_{i}$ with $1\le i \le 3$ be the standard basis vectors of $\mathbb{C}^{3}$. Let 
$$w=[1+f, g, h]^{\textrm{T}} := \pi^{-1}v_1$$
So $f,g,h$ are in $\mathbb{C}(\!(1/z)\!)$ and are $\mathcal O(z^{-1})$. Suppose now $V$ is stabilized by the vacuum operator $\nabla=\partial_{z}+\Lambda$. Since $z$ also preserves $V$, the following vector is in $V$ as well:
\begin{eqnarray}
\label{third-power-equation}
(\nabla^3-z)w=\begin{bmatrix}
f'''+3h'+g+3zh''+3zg'\\[3pt]
g'''+3f''+2h+3zh'\\[3pt]
h'''+3g''+3f'
\end{bmatrix}
\end{eqnarray}
Since all entries are $\mathcal O(z^{-1})$ and since $\pi$ is an isomorphism, it follows that the above expression is in fact the $0$ vector. Assume for contradiction that $h$ is non-zero and write $h=h_r z^{r}+h_{r-1}z^{r-1}+\cdots$ with $r<0$ and $h_{r}\ne 0$. For a non-zero Laurent series in $\mathbb{C}(\!(1/z)\!)$ we call the largest $i$ with non-zero coefficient of $z^{i}$ the degree and we denote the degree of $0$ to be $-\infty$. From the second entry in Equation (\ref{third-power-equation}) it follows that if $\textrm{deg}(g)\le r+2$ and $\textrm{deg}(f)\le r+1$ then $2 h_r+3r h_r=0$ with $r$ an integer, a contradiction. It then follows from the third entry of Equation (\ref{third-power-equation}) that both $f$ and $g$ are non-zero of degree bigger than $r$ and $\textrm{deg}(g)= \textrm{deg}(f)+1$. It then follows from the first entry of Equation (\ref{third-power-equation}) that if $g=g_{s} z^{s} + \cdots$ with $g_{s}\ne 0$, then $g_{s}+3sg_{s}=0$. This is a contradiction and hence $h= 0$. Similar arguments imply $g= 0$ and $f =0$. Since $\nabla w$ and $\nabla^2 w$ are in $V$ it follows that $v_i$ is in $V$ for all $i$ and hence so is $z^j v_i$ for all $j \ge 0$. It follows that $V=|0\rangle$, as desired. 

This type of elementary argument ceases to be as effective for Lie algebras with more complicated cyclic elements. For example for $\mathfrak e_6$ in one of the $27$-dimensional faithful representations, cyclic elements are up to conjugation and scaling of the form
\begin{gather*}
e_{2,1}+e_{3,2}+e_{4,3}+e_{5,4}+e_{6,4}+e_{7,5}+e_{7,6}+e_{8,6}+e_{9,7}+e_{10,7}+e_{10,8}+e_{11,9}+e_{12,9}+e_{12,10} +e_{13,11}+e_{14,11} +e_{14,12} 
\\+ e_{15,12}+e_{16,13} + e_{16,14}+e_{17,14} +e_{17,15}+e_{18,16}+e_{18,17}+e_{19,17}+e_{20,18}+e_{20,19}+e_{21,19}+e_{22,20}+e_{22,21}+e_{23,20 }\\ 
+e_{24,22} + e_{24,23}  +e_{25,24}+e_{26,25}+e_{27,26} + z \cdot ( e_{1,21}+e_{2,22}+e_{3,24}+e_{4,25}+e_{6,26}+e_{8,27}) 
\end{gather*}
To deal with these more complicated cases, we prove Theorem \ref{main-theorem} using the work by Cafasso and Wu \cite{CAW1} on the Witten-Kontsevich point for the Drinfeld-Sokolov hierarchy associated to $\mathfrak g$.

The second aim of the present work is to show that a certain perturbation of the vacuum operator also has meaning in the context of the Drinfeld-Sokolov hierarchy. To define this perturbation we note that since points in the Grassmannian are special subspaces of $\mathbb{C}(\!(1/z)\!)^{d}$, one can say that the hierarchy lives on a formal punctured disc with  coordinate $1/z$. This disc can be thought of as the disc around the point $z=\infty$ on the sphere $\mathbb{P}^{1}$. The vacuum operator can be viewed as a connection on this disc and as such has a canonical Katz extension to the sphere.  We give the precise definition in Section \ref{Katz-extension-section}.  It is a particularly nice extension, unique up to isomorphism due to its construction via a quasi-inverse functor, see \cite{KAT} (Equation 2.4.11).  The Katz extension,  also called canonical extension,  is not too singular near $z=0$, and has normal form properties in the style of Levelt-Turrittin,  similar to those enjoyed by connections on a disc. This notion is far removed from the topic of integrable hierarchies: Katz used it in \cite{KAT} for the calculation of differential Galois groups. Surprisingly, we show in Theorem \ref{Katz-extension-theorem} that in the case of the vacuum operator this connection has integrable meaning: It generates an additional symmetry of the hierarchy.

Note that Theorem \ref{main-theorem} is stated in the context of classical cyclic elements. In contrast, we are able to prove Theorem \ref{Katz-extension-theorem} for generalized Drinfeld-Sokolov hierarchies associated to arbitrary Heisenberg subalgebras of the affine Lie algebra $\widehat{\mathfrak g}$.  In order to state the result precisely we need to recall the notion of Kac coordinates associated to conjugacy classes in the Weyl group of $\mathfrak g$.

\subsection{Kac coordinates}

Kac and Peterson have shown in \cite{KP} that Heisenberg subalgebras of $\widehat{\mathfrak g}$ can be parameterized by conjugacy classes $[w]$ in the Weyl group of $\mathfrak g$. This parametrization can be (non-canonically) reduced to the case of primitive conjugacy classes and we assume from now on $[w]$ is indeed of this form. Let $r$ denote the rank of $\mathfrak g$, let $N$ denote the order of $w$, and let $\zeta_{N}$ be a primitive $N$'th root of unity. The starting point in the construction of the Heisenberg algebra $\mathcal H_{[w]}$ associated to $[w]$ is a collection of $r+1$ non-negative integers $\textbf{s}=(s_0,\cdots,s_r)$. To define them, one shows that there exist two Cartan subalgebras $\mathfrak h, \widetilde{\mathfrak h}$ of $\mathfrak g$ with the following properties:
\begin{enumerate}[(i)]
\item
There exists a finite-order inner automorphism $\sigma=\textrm{exp}(\textrm{ad }h)$ of $\mathfrak g$ that restricts to $w$ on $\mathfrak h$ and to the identity on $\widetilde{\mathfrak h}$
\item
Let $e_1,\cdots,e_r$ be generators of the simple root spaces with respect to $\widetilde{\mathfrak h}$ and let $e_{0}$ be a generator of the lowest root space.  There are coprime integers $s_i$ in $[0,N-1]$ such that
$$\sigma(e_{i}) = \zeta_{N}^{s_{i}} \cdot e_{i}$$
for all $i$ (for notational convenience we say that $0$ is coprime to all integers)
\end{enumerate}
The $s_i$'s are called the Kac coordinates of $[w]$ and are well defined up to permutations coming from diagram automorphisms. In general, there is also an ambiguity coming from choosing different liftings $\sigma$ of the Weyl group element $w$, but since we assume that $[w]$ is a primitive conjugacy class, the Kac coordinates turn out to be well defined.  We list below all regular primitive cases, with the ordering of simple roots as in \cite{BOURBAKI}. Most of these Kac coordinates are calculated by Bouwknegt in \cite{BOU}, see also \cite{REE}. For completeness we calculate the remaining cases in Section \ref{Katz-extension-section}.

\setlength{\tabcolsep}{8pt}

\renewcommand{\arraystretch}{1.2}

\begin{center}

\begin{tabular}{c|c}
\textrm{regular primitive conjugacy class  }& \textrm{Kac coordinates}\\ \hline \hline 
&\\
$\textrm{F}_4(a_1)$& $(1,0,1,0,1)$\\
$\textrm{E}_6(a_1)$& $(1,1,1,1,0,1,1)$ \\
$\textrm{E}_6(a_2)$& $(1,1,0,0,1,0,1)$\\
$\textrm{E}_7(a_1)$& $(1,1,1,1,0,1,1,1)$\\
$\textrm{E}_7(a_4)$& $(1,0,0,0,1,0,0,1)$\\ 
$\textrm{E}_8(a_1)$ & $(1,1,1,1,0,1,1,1,1)$ \\
$\textrm{E}_8(a_2)$ & $(1,1,1,1,0,1,0,1,1)$ \\
$\textrm{E}_8(a_3)$ & $(1,1,0,0,1,0,0,1,0)$\\
$\textrm{E}_8(a_5)$ & $(1,1,0,0,1,0,1,0,1)$ \\
$\textrm{E}_8(a_6)$& $(1,0,0,0,1,0,0,1,0)$\\
$\textrm{E}_8(a_8)$& $(1,0,0,0,0,1,0,0,0)$\\
$\textrm{D}_{2n}(a_{n-1})$ & $(1,1,0,1,0,1,0,1,\cdots,0,1,1)$\\
$\textrm{Coxeter class}$ & $(1,\cdots,1)$\\
&
\end{tabular}

\end{center}
The Kac coordinates are used in the construction of $\mathcal H_{[w]}$ by singling out a suitable twisted realization of the loop algebra:
Consider the type $\textbf{s}$-realization $\widehat{\mathfrak g}(\textbf{s})$ of $\widehat{\mathfrak g}$ together with the isomorphism
$$\Phi : \widehat{\mathfrak g} \longrightarrow \widehat{\mathfrak g}(\textbf{s}) $$
as described in \cite{KAC} (Chapter 8). The Heisenberg algebra $\mathcal H_{[w]}$ is then obtained as the inverse image via $\Phi$ of a subalgebra of $\widehat{\mathfrak g}(\textbf{s})$ built out of loops with suitable coefficients that are in particular in $\mathfrak h$. 

\subsection{Katz extensions}

It follows from the construction of $\mathcal H_{[w]}$ that the Kac coordinates $\textbf{s}$ yield a gradation on $\widehat{\mathfrak g}$ for which $\mathcal H_{[w]}$ has a homogeneous basis.  Fix an element $\Lambda$ in $\mathcal H_{[w]}$ that is of $\textbf{s}$-degree $1$ and introduce the vacuum operator
$$\nabla_{\textrm{vac}} =\partial_{z} +\Lambda$$
If $[w]$ is the conjugacy class of the Coxeter elements one can see that this recovers our earlier definition. 

Let $(-,-)$ be a non-degenerate symmetric invariant bilinear form on $\mathfrak g$ and via its restriction to the Cartan algebra $\widetilde{\mathfrak h}$ we view the fundamental weights $\omega_{i}$ as elements of $\widetilde{\mathfrak h}$. Let $\alpha_i$ denote the simple roots. We show that the Kac coordinates describe the Katz extension of the vacuum operators:

\begin{thm}
\label{Katz-extension-theorem}
Let $[w]$ be a regular primitive conjugacy class in the Weyl group of $\mathfrak g$, let $N$ be the order of $w$ and let $\textrm{\emph{\textbf{s}}}=(s_0,\cdots,s_r)$ be the Kac coordinates. The Katz extension of the vacuum operator is 
$$(\partial_{z} +\Lambda)^{\textrm{\emph{Katz}}}=\partial_{z} + \Lambda + \frac{1}{Nz} \cdot \sum_{i=1}^{r} s_{i} \cdot  \frac{2}{(\alpha_i,\alpha_i)} \cdot \omega_i $$

\end{thm}

As for Theorem \ref{main-theorem}, this result makes sense without any mention of integrable hierarchies.  However, as indicated earlier, our motivation comes from showing that these Katz extensions of vacuum operators do in fact play a crucial role in integrable systems. Delduc and Feher describe in \cite{DF} how to construct for every suitable element $\Lambda$ in $\mathcal H_{[w]}$ of $\textbf{s}$-degree $1$ an integrable hierarchy of Drinfeld-Sokolov type (the original construction of Drinfeld and Sokolov in \cite{DS} is recovered by letting $[w]$ be the conjugacy class of a Coxeter element).  See also the work of de Groot-Hollowood-Miramontes \cite{GHM} for another construction of the generalized Drinfeld-Sokolov hierarchies.  As shown by Hollowood-Miramontes-Guillen in \cite{HMG}, these hierarchies possess additional symmetries which commute with the flows but amongst themselves obey commutation relations of Virasoro type.  As a consequence of Theorem \ref{Katz-extension-theorem} we will show:

\begin{cor}
The Katz extension of the vacuum operator $\partial_{z} +\Lambda$ describes an additional symmetry of the generalized Drinfeld-Sokolov hierarchy associated to $\mathfrak g$ and $\Lambda$.
\end{cor}

This shows that the a priori unrelated notion of the Katz extension does indeed play a crucial role in this integrable hierarchy.  The proof is by direct calculation: We calculate the Katz extension and compare it to the known shape of additional symmetries. This leaves open a more conceptual understanding between the two notions and a more direct understanding of this relation would be desirable!

\section{Characterizing the vacuum} 
\label{vacuum-section}

In this section we prove Theorem \ref{main-theorem}. Write $\gamma=\textrm{exp}(-U)$  where $U=\sum_{i<0} U_{i}$ with $U_{i}$ in $\widehat{\mathfrak g}$ of principal degree $i$, and as mentioned before assume that $U$ is in $z^{-1} \mathfrak g[z^{-1}]$.  Suppose the point $V =\gamma |0\rangle$ of the Grassmannian is stabilized by the vacuum operator:
\begin{eqnarray}
\label{V-stabilization-equation}
(\partial_{z} + \Lambda) V\subseteq V
\end{eqnarray}
Equivalently, $|0\rangle$ is stabilized by $\gamma^{-1}(\partial_{z} +\Lambda)\gamma$. Let us rewrite this stabilization condition as a vanishing result. Expand
\begin{eqnarray}
\label{Ai-equation}
\gamma^{-1}(\partial_{z} +\Lambda)\gamma = \textrm{exp}(\textrm{ad }U)(\partial_{z} + \Lambda)=\partial_{z} + \Lambda + \sum_{i\le 0} A_{i} z^{i}
\end{eqnarray}
with $A_{i}$ in $\mathfrak g$ independent of $z$. Then the stabilization is equivalent to the vanishing result
\begin{eqnarray}
\label{vanishing-equation}
A_{i}=0 \;\;\; \textrm{ for all $i<0$}
\end{eqnarray}
In particular, the validity of Theorem \ref{main-theorem} is independent of the choice of faithful representation $\xi$ of $\mathfrak g$. We now show how to deduce the theorem from various results by Cafasso and Wu in \cite{CAW1}, concerning the existence and uniqueness of the Witten-Kontsevich point in Drinfeld-Sokolov phase space. This point can be characterized by various versions of what is called the string equation, and Equation (\ref{V-stabilization-equation}) turns out to be closely related.

For $U$ as before,  following standard notation in the setting of integrable systems,  we call elements of the form $\textrm{exp}(-U)$ gauge transformations.  Consider the Drinfeld-Sokolov hierarchy associated to $\mathfrak g$. This considers the choice of a cyclic element $L$, and we choose $L=-\Lambda$. Cafasso and Wu show in \cite{CAW1} (Theorem 3.11) the existence and uniqueness of a Witten-Kontsevich point in the phase space. Write it as 
$$W = \mu |0\rangle $$
for a gauge transformation $\mu$.  Let $h$ denote the Coxeter number of $\mathfrak g$ and let $\rho^{\vee}$ be half the sum of positive co-roots, viewed as an element of $\widetilde{\mathfrak h}$.  By \cite{CAW1} (Lemma 3.9)
\begin{eqnarray}
\label{special-gauge-equation}
\mu^{-1} \left (\partial_{z}  +\frac{\rho^{\vee}}{h} \cdot \frac{1}{z} +\Lambda \right ) \mu =\partial_{z} +\Lambda
\end{eqnarray}
By Equation (\ref{vanishing-equation}) it follows that
\begin{eqnarray}
\label{W-string-equation}
 \left (\partial_{z}  +\frac{\rho^{\vee}}{h} \cdot \frac{1}{z} +\Lambda \right ) (\mu\gamma |0\rangle ) \subseteq \mu \gamma |0\rangle 
 \end{eqnarray}
 
Equation (\ref{W-string-equation}) is nothing but the string equation for the dressing operator, \cite{CAW1} (Equation 3.29), when all flow variables are set to $0$.  Let $t_{i}$ be the flow variables, indexed by the positive exponents of $\widehat{\mathfrak g}$, and let the $\Lambda_{i}$'s be principal degree $i$ basis elements of the principal Heisenberg algebra as in. Since
$$\exp\left (-\sum_{i} t_i \Lambda_i \right )( \partial_{z}  +\frac{\rho^{\vee}}{h} \cdot \frac{1}{z}) \exp \left (\sum_{i} t_i \Lambda_i \right )=  \partial_{z}  +\frac{\rho^{\vee}}{h} \cdot \frac{1}{z} - \sum_{i} \frac{i t_i}{h}  \Lambda_{i-h}$$
it follows that \cite{CAW1} (Equation 3.29) also holds for non-zero time variables $t_i$.  Put differently, $\mu \gamma |0\rangle$ is the Witten-Kontsevich point. By the uniqueness of the Witten-Kontsevich point it follows that
$$\mu \gamma |0\rangle = \mu |0\rangle $$
Therefore $\gamma=1$ and $ V= |0\rangle$, as desired. This completes the proof of Theorem \ref{main-theorem}.

\begin{remark}
Some of the arguments by Cafasso and Wu in \cite{CAW1} depend on the construction and properties of tau functions associated to each point in Drinfeld-Sokolov phase space. Such arguments can be subtle since the very definition of the tau function has a certain non-uniqueness and furthermore, arguments that involve a shifting of time variables have to be treated very carefully since in general the tau functions are only formal power series in the time variables.

We therefore outline an algorithm that sidesteps the tau function arguments. The starting point is to calculate
$\textrm{exp}(\textrm{ad } U)(\partial_{z} + \Lambda)$
recursively with respect to principal degree. Equation (\ref{vanishing-equation}) then puts constraints on the $U_{i}$'s. The crucial point is that one can show that these equations force $A_{0}=0$ and hence 
\begin{eqnarray}
\label{vacuum-characterization}
\gamma^{-1}(\partial_{z} + \Lambda)\gamma=\partial_{z}+\Lambda
\end{eqnarray}
This in fact implies the theorem:  By \cite{CAW1} (Lemma 3.9, Theorem 3.11) there is a unique gauge transformation $\mu$ such that
\begin{eqnarray}
\label{KS-to-vacuum-equation}
\mu^{-1}(\partial_{z} +\frac{\rho^{\vee}}{h} \cdot \frac{1}{z} +\Lambda) \mu = \partial_{z} + \Lambda
\end{eqnarray}
This result does not depend on tau function arguments. It now follows from Equation \ref{vacuum-characterization} that $\gamma=1$, as desired.

We illustrate the algorithm for $\mathfrak s\mathfrak l_{2}$. Consider a standard basis $E,F,H$. The elements $U_{i}$ are given by
 $$U_{-1}=   \frac{a_1}{z} \cdot E \; ,\;   U_{-2}=                                                                                          
\frac{a_2}{z} \cdot  H \;, \;  U_{-3}=                                                                     \frac{a_3}{z^2} \cdot E + \frac{a_4}{z} \cdot F \; ,\;
U_{-4}=                                                                                          
\frac{a_5}{z^2} \cdot H, \cdots$$
where the $a_{j}$'s are arbitrary constants. Note that $U_{-1}$ has no contribution from $F$ since by assumption all $U_{i}$'s are not only of principal degree $i$ but also of negative $z$-degree. One calculates
$$A_0=a_{1} \cdot H -2a_{2} \cdot F$$
The equations up to principal degree $-5$ yield
\begin{eqnarray*}
0&=& a_1^2 - 2 a_2  \\
0 &=& a_1  a_2 -a_3 + a_4 \\
0 &=&  a_1 a_4 +2 a_2^2- 2a_5 \\
 0 &=& a_1a_4+2 a_2^2  -2a_1a_3   - a_1 + 2a_5 
\end{eqnarray*}
These equations force $a_1=a_2=0$ and hence $A_0=0$, as desired.

The algorithm can be seen to work for all simple Lie algebras. Even for, say, $\mathfrak e_6$ the resulting system of equations becomes quite involved, so we do not list them here.

\end{remark}

\section{Katz extension and additional symmetries}
\label{Katz-additional-section}

In his work on the calculation of differential Galois groups, Katz introduced in \cite{KAT} the notion of the canonical extension of a connection on a formal disc to a connection on the sphere. One useful property in this context is that the local Galois group of a connection agrees with the global Galois group of the canonical Katz extension. These notions are far removed from the Drinfeld-Sokolov hierarchies but we show, somewhat surprisingly, that the Katz extension of the vacuum operator is an important object for the associated integrable hierarchy.

\subsection{Additional symmetries}

It is known, see for example \cite{DF} and \cite{GHM}, that for suitable graded elements $\Lambda$ in a Heisenberg subalgebra of an affine Lie algebra there exists an associated generalized Drinfeld-Sokolov hierarchy. The Heisenberg algebras are classified by Kac and Peterson in \cite{KP}. Up to conjugacy they are parameterized by conjugacy classes in the Weyl group of $\mathfrak g$. Equivalently, this yields a classification of maximal abelian subalgebras of loop algebras $\widehat{\mathfrak g}=\mathfrak g[z,z^{-1}]$ and it is for this variant that we now recall relevant aspects of the constructions.
 
Let $w$ denote a Weyl group element and let $\sigma = \textrm{exp}(\textrm{ad} \; h)$ be a finite-order inner automorphism of $\mathfrak g$ that restricts to $w$ on a Cartan algebra $\mathfrak h$.  Let $a_{0},\cdots,a_{r}$ be the Kac labels of $\mathfrak g$ (see \cite{KAC} Table Aff $1$ in Chapter 4). Note that $a_0=1$ and the highest root of $\mathfrak g$ is written in terms of the simple roots $\alpha_i$ as $\sum_{i=1}^{r} a_i \alpha_i$. Kac has shown that, possibly after conjugation, one can assume that $h$ is in a Cartan algebra $\widetilde{\mathfrak h}$ and there is a choice of generators $\{e_i,f_i,h_i \}_{1\le i \le r}$ of $\mathfrak g$ satisfying the Serre relations, the $h_i$'s  span $\widetilde{\mathfrak h}$, and
$$\sigma(e_i)=\exp(\frac{2 \pi i s_i}{\sum_{j=0}^{r} a_j s_j}) \cdot e_i$$
where the $s_i$'s are non-negative integers and the non-zero $s_i$'s are co-prime. In the following, we let
\begin{eqnarray}
\label{N-equation}
N:=\sum_{i=0}^{r} a_i s_i
\end{eqnarray}
The $s_i$'s are called the Kac coordinates of $\sigma$ and we define $\textbf{s}=(s_0,\cdots,s_r)$. Fix a primitive $N$'th root of unity $\zeta_{N}$ and for any integer $i$ denote by $\mathfrak g_{[i]}$ the subspace of $\mathfrak g$ on which $\sigma$ acts by multiplication by $\zeta_{N}^{i}$. Consider the Lie algebra
$$\widehat{\mathfrak g}(\sigma)= \bigoplus_{i}  \mathfrak g_{[i]} \cdot z^{i}$$
For each root $\alpha =\sum_{i} k_i \alpha_i$ define its degree as $\textrm{deg } \alpha= \sum_{i} k_{i} s_{i}$. As discussed in \cite{KAC} (Chapter 8), there is an isomorphism of Lie algebras
\begin{eqnarray}
\label{Phi-equation}
\Phi : \widehat{\mathfrak g} \rightarrow \widehat{\mathfrak g}(\sigma)
\end{eqnarray}
that on a root space $\mathfrak g_{\alpha}$ is given by multiplication by $z^{\textrm{deg }\alpha}$. More generally, on $\mathfrak g_{\alpha} z^{k}$ the map $\Phi$ is given by
\begin{eqnarray}
\label{Phi-definition}
x z^{k} \mapsto x z^{\textrm{deg }\alpha + kN}
\end{eqnarray}
Consider the subalgebra of $\widehat{\mathfrak g}(\sigma)$ given by
$$\widehat{\mathfrak h}(\sigma) = \bigoplus_{i} \mathfrak h_{[i]} \cdot z^{i}$$
where $\mathfrak h_{[i]} = \mathfrak h \cap \mathfrak g_{[i]}$. The Heisenberg algebra $\mathcal H_{[w]}$ associated by Kac and Peterson to the conjugacy class of $w$ is then given as
$$\mathcal H_{[w]}=\Phi^{-1} \left ( \widehat{\mathfrak h}(\sigma) \right )$$ 
Since for all integers $k$ clearly $\mathfrak h_{[i + kN]}=\mathfrak h_{[i]}$ it follows that one can choose indices $i_t$ with $0 \le i_{1} \le \cdots \le  i_{r} < N$ and elements $\mu_{i_{t}}$ in $\mathfrak h_{[i_{t}]} \cdot z^{i_{t}}$ such that 
\begin{eqnarray}
\label{basis-equation}
\mu_{i_{t}} \cdot z^{i_{t}+kN} \;\;\; \textrm{with } 1\le t \le r  \;\; , \;\;  k \in \mathbb{Z}
\end{eqnarray}
is a basis of $\widehat{\mathfrak h}(\sigma)$. Let $\textrm{E}$ denote the multi-set $\{i_t + kN \; | \;  1\le t \le r  \;, \; k \in \mathbb{Z} \}$ and let $\textrm{E}^{\ge 0}$ denote the multi-set corresponding to $k \ge 0$. The above basis of $\widehat{\mathfrak h}(\sigma)$ can then be labeled as $\{\lambda_{j} \}_{j \in \textrm{E}}$ such that $\lambda_{j}$ is in $\mathfrak h_{[j]} \cdot z^{j}$ and
\begin{eqnarray}
\label{derivative-equation}
z^{1-N} \partial_{z} \lambda_{j}= j \lambda_{j-N}
\end{eqnarray}
Define the corresponding basis of $\mathcal H_{[w]}$ as
\begin{eqnarray}
\label{Lambda-definition}
\Lambda_j := \Phi^{-1} \left(\lambda_{j}\right )
\end{eqnarray}
Let $\theta$ denote the highest root, let $E_{\theta}$ and $E_{-\theta}$ be certain non-zero elements in the highest and lowest root space of $\mathfrak g$, see \cite{KAC} (\S 7.4) for the precise definition. Let $e_0=zE_{-\theta}$ and let $f_0=E_{\theta}/z$. 
For a collection $d_0,\cdots,d_r$ of non-negative integers (not all zero), define a $\mathbb{Z}$-gradation on $\widehat{\mathfrak g}$ by
$$\textrm{deg } e_i =d_i \; \;  ,\; \; \textrm{deg } f_i = -d_i$$
for $0\le i \le r$. If one chooses $d_i=s_i$ for all $i$ then $\Lambda_j$ is of degree $j$, see \cite{KAC}. We call this the $\textbf{s}$-gradation. 
For example $\textbf{s}_{0}:=(1,0,\cdots,0)$ yields the standard homogeneous gradation. 

As in \cite{DF} we now assume $\Lambda$ in $\mathcal H_{[w]}$ is of $\textbf{s}$-degree $1$,  where the gradation $\textbf{s}$ satisfies $\textbf{s} \ge \textbf{s}_{0}$ (hence the first entry of $\textbf{s}$ is non-zero). Let $x$ be an indeterminate. The Lax operator of the corresponding generalized Drinfeld-Sokolov hierarchy is of the form
$$L=\partial_{x}+\Lambda + q$$ 
where $q$ is in $\widehat{\mathfrak g}_{\ge 0}(\textbf{s}_{0})$ and in $\widehat{\mathfrak g}_{<1}(\textbf{s})$, meaning $q$ in $\widehat{\mathfrak g}$ is of non-negative $\textbf{s}_{0}$-degree and of $\textbf{s}$-degree less than $1$.

For a Lie algebra element $M$ let $M^{+}$ denote the non-negative part of $M$ in the $\textbf{s}$-gradation and $M^{-}=M-M^{+}$ the negatively graded part. The simple observation for defining flows on operators of the form $L$ is that if $M$ satisfies $[M,L]=0$ then
$$[M^{+},L]=-[M^{-},L]$$
is both in $\widehat{\mathfrak g}_{\ge 0}(\textbf{s}_{0})$ and $\widehat{\mathfrak g}_{<1}(\textbf{s})$. Therefore one can define a flow, with flow variable $t_{M}$ say, by
$$\partial_{t_{M}} L = [M^{+},L]$$
Since $L$ can be gauge transformed to an operator of the form $\partial_{x} + H$ with $H$ in the Heisenberg algebra (cf. \cite{GHM} Proposition 3.2), it follows that gauge transforms of elements in $\mathcal H_{[w]}$ can play the role of $M$. This yields the usual flows of the hierarchy. The fact that the elements in $\mathcal H_{[w]}$ commute amongst themselves implies that the flows commute.

The notion of additional symmetry of the hierarchy is obtained by relaxing in the equation $[M,L]=0$ the requirement that $M$ is in the Lie algebra. Rather, it can now involve also the derivations acting on $\widehat{\mathfrak g}$. Suppose there is a derivation $d$ such that for all $k$ in $\textrm{E}^{\ge 0}$
\begin{eqnarray}
\label{derivation-equation}
\left [d,\Lambda_{k} \right ]=\frac{k}{N} \cdot  \Lambda_{k-N}
\end{eqnarray}
Let $t_{k}$ denote the flow variable associated to the Heisenberg algebra element $\Lambda_{k}$. Then for all $k$ as above
\begin{eqnarray}
\label{commutation-equation}
\left [d+ \frac{1}{N} \sum_{j \in \textrm{E}^{\ge 0}} j t_j \Lambda_{j-N}  \; , \;  \partial_{t_{k}} - \Lambda_{k} \right ]=0
\end{eqnarray}
In analogy with the construction of the flows, this allows to define an evolution equation for the Lax operator $L$: One defines
$$\partial_{\beta} L = [L,  \delta_{<0}]$$
where $\delta_{<0}$ is obtained from $d+ \frac{1}{N} \sum_{j \in \textrm{E}^{\ge 0}} j t_j \Lambda_{j-N}$ via a dressing procedure.  We refer to the work of Hollowood-Miramontes-Guillen \cite{HMG} and Wu \cite{WU} (Equation 4.26, note that we use slightly different sign conventions) for details. This evolution commutes with the flows, here Equation (\ref{commutation-equation}) is the crucial ingredient.  See \cite{HMG} (Equation 3.11) and \cite{WU} (Proposition 4.5).  One obtains an additional symmetry of the hierarchy.  It is not called a flow itself since it is part of an infinite family of analogous symmetries that all commute with the flows but amongst themselves obey commutation relations of Virasoro type. 

The operator $d$ can be viewed as the generator of the symmetry and we now describe how to find such a derivation.  Let $\lambda= \Phi(\Lambda)$. Consider on $\widehat{\mathfrak g}(\sigma)$ the derivation 
$$ D:= z^{1-N}\cdot \partial_{z} + N \lambda$$ 
where $\lambda$ stands for the derivation $[\lambda,-]$. We now calculate what this derivation looks like on $\widehat{\mathfrak g}$.

Let as before $\{e_i,f_i,h_i\}_{i=1,\cdots,r}$ be a set of generators of $\mathfrak g$ satisfying the Serre relations. For all $1 \le i \le r$ one has
$$\Phi(e_{i}) = z^{s_{i}} e_{i}$$
Since $\theta = \sum_{i=1}^{r} a_{i} \alpha_{i}$ it follows from Equation (\ref{Phi-definition}) that
$$\Phi(e_0)=z^{s_{0}} E_{-\theta}$$
It now follows for $1\le i \le r$ that
$$(\Phi^{-1} \circ D \circ  \Phi) (e_{i})=s_{i} \Phi^{-1} (z^{s_{i}-N} e_{i}) + [N\Lambda,e_i]=\frac{s_{i} e_{i}}{z} + [N\Lambda,e_i]$$
Since $\Phi(E_{-\theta})=z^{s_0-N}E_{-\theta}$ one also has
$$(\Phi^{-1} \circ D \circ \Phi) (e_{0})=s_{0}\Phi^{-1} (z^{s_{0}-N} E_{-\theta})+[N\Lambda, e_0]=\frac{s_{0}e_{0}}{z} +[N\Lambda, e_0] $$
We now describe $\Phi^{-1} \circ D \circ \Phi$ in a slightly different manner.
For each $1\le i \le r$ let $\omega_i$ be the $i$'th fundamental weight, viewed via a choice $(-,-)$ of non-degenerate invariant bilinear form as an element of $\widetilde{\mathfrak h}$. So 
$$[\omega_i,e_j]=\delta_{i,j}\cdot  \frac{(\alpha_j,\alpha_j)}{2} \cdot e_{j}$$
Recall that for the Weyl group element $w$ that we started with we have chosen a lift $\sigma=\textrm{exp}(\textrm{ad }h)$ for some $h$ in $\widetilde{\mathfrak h}$. Rescale the element $h$ as $h = 2 \pi i R$, so that 
$$[R,e_{i}]=\frac{s_{i}}{N}e_{i}$$
Concretely, $R$ is given by 
\begin{eqnarray}
\label{r-Equation}
R=\frac{1}{N}\cdot \sum_{i=1}^{r} \frac{2}{(\alpha_i,\alpha_i)} s_{i} \omega_i
\end{eqnarray}
Then for $i\ge 1$
$$\left [N (\partial_{z}+\frac{R}{z} +\Lambda), e_{i} \right ]=\frac{s_{i}}{z} e_{i} + [N\Lambda,e_i]$$
For $i=0$ one has
$$\left [N(\partial_{z}+\frac{R}{z} +\Lambda), e_{0} \right ]=N E_{-\theta}-\sum_{i=1}^{r} a_{i} s_{i} E_{-\theta} + [N\Lambda,e_0]=s_{0} E_{-\theta}+ [N\Lambda,e_0]$$
It follows that the derivation $\Phi^{-1} \circ D \circ \Phi$ agrees with the derivation $N(\partial_{z}+\frac{R}{z} + \Lambda)$ on a set of generators and hence the two derivations are equal. Hence
$$N(\partial_{z}+\frac{R}{z} +  \Lambda)(\Lambda_{j})=j \Lambda_{j-N}$$
Therefore, for the desired derivation $d$ in Equation (\ref{derivation-equation}) one can take 
$$d=\partial_{z}+\frac{R}{z} + \Lambda$$ 
See for example \cite{WAK} for a more in-depth discussion of these kind of derivations.

\subsection{Katz extension}
\label{Katz-extension-section}

In Section \ref{vacuum-section} we have shown in what manner the vacuum operator $\partial_{z} + \Lambda$ characterizes the Drinfeld-Sokolov vacuum $|0\rangle$.  In the Grassmannian formulation,  the hierarchy corresponds to certain subspaces of a $\mathbb{C}(\!(1/z)\!)$-vector space and in this sense the hierarchy lives on a punctured disc around $z=\infty$.  It is then natural to view the vacuum operator as a connection on this disc.  After embedding the disc into a sphere,  one can ask what the canonical Katz extension of the vacuum operator is.  We answer this question in the current section. 

Let us recall the notion of the Katz extension. Let $t$ be an indeterminate. A connection $\nabla$ on a formal punctured disc $\textrm{Spec }\mathbb{C}(\!(t)\!)$ is by definition a $\mathbb{C}$-linear endomorphism $\nabla$ of a finite-dimensional $\mathbb{C}(\!(t)\!)$-vector space $V$ that satisfies the Leibniz identity 
$$\nabla(f(t)\cdot v) =(\partial_{t} f(t))\cdot v+f \cdot \nabla(v)$$ 
for all $f$ in $\mathbb{C}(\!(t)\!)$ and $v$ in $V$. After choosing a basis, the connection corresponds to a linear differential operator 
$$\nabla=\partial_{t} + A$$
with $A$ in $\textrm{GL}_{n}\mathbb{C}(\!(t)\!)$ where $n=\textrm{dim}_{\mathbb{C}(\!(t)\!)} V$.

Starting with work of Levelt \cite{LEV} and Turrittin \cite{TUR} in the 1950's, it is known that connections on a formal disc possess a normal form resembling the Jordan canonical form. In general, this normal form is only achieved after extending scalars to a finite extension of $\mathbb{C}(\!(t)\!)$: Let $r\ge 1$ be an integer and let $s=t^{1/r}$. The pull-back of a connection $(V,\nabla)$ over $\mathbb{C}(\!(t)\!)$ to a connection over $\mathbb{C}(\!(s)\!)$ has underlying $\mathbb{C}(\!(s)\!)$-vector space simply $V \otimes_{\mathbb{C}(\!(t)\!)} \mathbb{C}(\!(s)\!)$. The analogue of $\nabla$ is $[r]^{*}\nabla$ which satisfies in particular
$$[r]^{*}\nabla (v \otimes 1)= v \otimes r s^{r-1} $$
The results of Levelt and Turrittin imply that for a suitable $r$, the connection $[r]^{*}\nabla$ is isomorphic to a successive extension of one-dimensional connections in an essentially unique manner.

For connections on the sphere the analogous statement fails in general. The idea of the canonical Katz extension is to find nonetheless a special connection on the sphere that restricts to $\nabla$ while simultaneously exhibiting behavior of Levelt-Turrittin type. Let us give a precise definition in the case where after a suitable pull-back $\nabla$ is the direct sum of one-dimensional connections. This will be the case of interest in our application to vacuum operators of Drinfeld-Sokolov hierarchies. 

Given a connection $\nabla=\partial_{t} + A$ on the disc $\textrm{Spec } \mathbb{C}(\!(t)\!)$, an extension to the sphere is simply $\nabla'=\partial_{t} +B$ with $B$ in $\textrm{GL}_{n}\mathbb{C}[t,t^{-1}]$ such that $\nabla$ and $\nabla'$ are isomorphic on $\textrm{Spec } \mathbb{C}(\!(t)\!)$: There is $g$ in $\textrm{GL}_{n}\mathbb{C}(\!(t)\!)$ such that
$$g\partial_{t}g^{-1} + gAg^{-1}=B$$
Given such an extension $\nabla'$ one can consider its local behavior around the disc $t=\infty$: Under the coordinate change $z=1/t$ the connection $\partial_{t} +B$ changes as
$$\partial_{t}+B \leadsto \partial_{z} -\frac{1}{z^{2}}\cdot B$$
If there is a gauge transformation in $\textrm{GL}_{n}\mathbb{C}(\!(z)\!)$ that puts this operator into the form 
$$\partial_{z} + C_{-1} \frac{1}{z}  + C_{0} + C_{1} z+ \cdots $$
with each $C_i$ constant, then we call the connection regular singular around $z=0$. A canonical Katz extension $\nabla^{\textrm{can}}$ of $\nabla$ is by definition an extension of $\nabla$ to the sphere that is regular singular around $z=0$ and that after a finite pull-back is a direct sum of one-dimensional connections.

As before, we now choose a faithful complex representation $\xi : \mathfrak g \rightarrow \mathfrak g\mathfrak l(V)$ and we usually simply write $x$ for $\xi(x)$. If now $A$ is in $\mathfrak g[z,z^{-1}]$ then via $\xi$ one can view $\partial_{z} + A$ as a connection on the formal disc around $z=\infty$. In the case where $A$ is a suitable element in a Heisenberg algebra we calculate in Theorem \ref{Katz-extension-theorem} the Katz extension.

\begin{proof}(of Theorem \ref{Katz-extension-theorem})

Let $[w]$ be a regular primitive conjugacy class in the Weyl group. Let $N$ denote the order of $w$ and fix a primitive $N$'th root of unity $\zeta$. Fix a Cartan algebra $\mathfrak h$ in $\mathfrak g$ and view the Weyl group via its reflection representation on the corresponding space of roots. Let us recall why the Kac coordinates associated to primitive Weyl group elements are uniquely determined. A conjugacy class $[w]$ is called primitive if in the reflection representation
$$\textrm{det}(1-w)=\textrm{det}(A)$$
where $A$ is the Cartan matrix of $\mathfrak g$. In particular, this determinant is non-zero and hence $1$ is not an eigenvalue of $w$. It follows that $w$ is an elliptic element. The desired uniqueness is then implied for example by \cite{AHN} (Lemma 1.1.3): Suppose for $i=1,2$ there is a lift $\sigma_{i}=\textrm{exp}(\textrm{ad }g_i)$ of $w$ to an inner automorphism of $\mathfrak g$. Then 
$$\sigma_2 = \textrm{exp}(\textrm{ad }h) \sigma_1$$
for some $h$ in $\mathfrak h$. Since $w$ is elliptic, the kernel of the endomorphism of $\mathfrak h$ given by $t \mapsto t - w(t)$ is $0$. Hence the map is surjective and there exists $t$ in $\mathfrak h$ such that
$$t-w(t)=h$$
Then
\begin{eqnarray*}
\textrm{exp}(\textrm{ad }t) \sigma_{1} \textrm{exp}(-\textrm{ad }t)&=& \textrm{exp}(\textrm{ad }t) \textrm{Ad}_{\textrm{exp}(g_1) \textrm{exp}(-t)\textrm{exp}(-g_1) \textrm{exp}(g_1) } \\
&=&\textrm{exp}(\textrm{ad }t) \textrm{Ad}_{\textrm{exp}(w(-t))} \textrm{exp}(\textrm{ad }g_1) \\
&=&\textrm{exp}(\textrm{ad }(t-  w(t))) \textrm{exp}(\textrm{ad }g_1) \\
&=& \sigma_2
\end{eqnarray*}
Hence the conjugacy class of lifts of $w$ is well defined and the Kac coordinates are unique.

Using this uniqueness together with \cite{REE} (Proposition 2.2) it follows that the order $N$ of $w$ is also equal to $\sum_{i=0}^{r} a_i s_i$, where the $s_i$'s are the Kac coordinates of $\sigma$ and as before the $a_i$'s are the Kac labels of $\mathfrak g$. Hence our notation is consistent with Equation (\ref{N-equation}). By definition, $[w]$ being regular means that there is an eigenvector $\lambda$ in the Cartan algebra $\mathfrak h$ such that $(\lambda,\alpha)\ne 0$ for all roots $\alpha$. Put differently:
\begin{eqnarray}
\label{Cartan-property}
\textrm{Cent}_{\mathfrak g}(\lambda)=\mathfrak h
\end{eqnarray}
By \cite{SPR} (Theorem 4.2) the eigenvalue corresponding to $\lambda$ is a root of unity of the same order as $w$. Furthermore, by loc. cit. (Proposition 4.7), if $i$ is co-prime to the order of $w$ then $w^{i}$ is conjugate to $w$. Hence, possibly after moving within the conjugacy class $[w]$, we can assume that the eigenvalue is our chosen primitive root of unity $\zeta$. Let $\textbf{s}=(s_0,\cdots,s_r)$ be the Kac coordinates of $[w]$. As in Equation (\ref{r-Equation}) define
$$r_{\textbf{s}}=\frac{1}{N} \cdot \sum_{i=1}^{r} \frac{2}{(\alpha_i,\alpha_i)} s_{i} \omega_i$$
Suppose now that the following holds:
\begin{condition}
\label{first-condition}
$N-1$  is the largest eigenvalue of $ \textrm{ad } N r_{\textbf{s}} $ on $\mathfrak g$
\end{condition}
Previously, we denoted by $\mathfrak g_{[i]}$ the subspace of $\mathfrak g$ on which $\sigma$ acts by $\zeta^{i}$. Now define a $\mathbb{Z}$-gradation on $\mathfrak g$ by putting for all $1\le i\le r$
$$\textrm{deg } e_{i} = s_{i} \;\; ,\;\; \textrm{deg } f_{i} = -s_{i} \;\; ,\;\; \textrm{deg } h_{i} = 0$$    
and denote by $\mathfrak g_{i}$ the degree $i$ space. Since for all $i$ one has $\dim \mathfrak g_{i}= \dim \mathfrak g_{-i}$ it follows from Condition \ref{first-condition} that $1-N$ is the smallest eigenvalue of $ \textrm{ad } N r_{\textbf{s}} $ on $\mathfrak g$. The only two integers in $[1-N,N-1]$ congruent to $1$ modulo $N$ are $1$ and $1-N$ and therefore
$$\lambda = a_{1}+a_{1-N}$$ 
for $a_{i}$ in $\mathfrak g_{i}$. Similar to a previous calculation one sees that for $$\Lambda = a_{1} + z \cdot a_{1-N}$$ 
in $\widehat{\mathfrak g}$ one has $\Phi(\Lambda)=\lambda\cdot z$. 
We want to show that there is $\nu$ in $\mathfrak g$ solving the equation
\begin{eqnarray}
\label{crucial-equation}
[\nu,\Lambda]=r_{\textbf{s}}
\end{eqnarray}
The existence of $\nu$ follows from a case by case analysis (with respect to the Weyl group conjugacy class $[w]$) from results due to various authors and collected by Delduc and Feher in \cite{DF}. In the following we summarize the relevant arguments. The idea is that no matter what $a_1$ is, for every element $a_{0}$ in $\mathfrak g_{0}$ (and hence in particular for $a_0=r_{\textbf{s}}$), there is $a_{-1}$ in $\mathfrak g_{-1}$ such that $[a_{-1},a_{1}]=a_0$.

Since $\lambda$ is regular semi-simple it follows that the kernel of $\textrm{ad } (a_{1}+a_{1-N})$ has $0$ intersection with $\mathfrak g_{<0}$ since the elements of the latter space are nilpotent. Since $1-N$ is the smallest eigenvalue this implies that the above kernel is in fact equal to the kernel of $\textrm{ad } a_1$. Hence $[a_{1},\mathfrak g_{-1}]$ and $\mathfrak g_{-1}$ have the same dimension. Suppose one knows:
\begin{condition}
\label{second-condition}
$\dim \mathfrak g_{-1} = \dim \mathfrak g_{0}$
\end{condition}

If Condition 2 holds, since $[a_{1},\mathfrak g_{-1}] \subseteq \mathfrak g_{0}$, it follows that $
 [a_{1},\mathfrak g_{-1}]$ is all of $\mathfrak g_{0}$. In particular, there exists $\nu
 $ in $\mathfrak g_{-1}$ such that $[\nu, a_{1}]=r_{\textbf{s}}$. Since $[a_{1-N},\mathfrak 
 g_{<0}]=0$ it follows that that
$$[\nu,a_{1}+z\cdot a_{1-N}]=r_{\textbf{s}}$$
and Equation (\ref{crucial-equation}) does indeed have a solution.
 
The fact that for all primitive regular conjugacy classes Condition \ref{first-condition} and Condition \ref{second-condition} do hold is part of what is shown in \cite{DF} (Appendix A), using the results of various authors. Since not all details are given in loc. cit. we give some illustrative sample calculations below.

\begin{example} Consider the conjugacy class of a Coxeter element, hence $N=h$ is th Coxeter number. Kostant showed in \cite{KOS} that the Kac coordinates are $\textbf{s}=(1,\cdots,1)$.
Write the highest root $\alpha_{\textrm{max}}$ as $\alpha_{\textrm{max}} = \sum_{i=1}^{r} a_{i} \alpha_{i}$. Then $\sum_{i=1}^{r} a_{i} = h-1$ and therefore
$$\alpha_{\textrm{max}}\left (hr_{\textbf{s}} \right )=\sum_{i=1}^{r}a_{i}=h-1$$
and hence Condition \ref{first-condition} holds. In terms of a basis $\{e_{i},f_{i},h_{i}\}_{1\le i \le r}$ of $\mathfrak g$ satisfying the Serre relations one sees that $\mathfrak g_{0}$ is the span of the $h_{i}$'s and $\mathfrak g_{-1}$ is the span of 
the $f_{i}$'s. Hence 
$$\dim \mathfrak g_{-1} =r= \dim \mathfrak g_{0}$$
and Condition \ref{second-condition} holds.
\end{example}

\begin{example}
Consider the conjugacy class $\textrm{E}_6(a_1)$. The order of the Weyl group element is $9$, see for example \cite{BOU} (Table I) for orders of all primitive Weyl group elements. We follow the ordering of simple roots as in \cite{BOURBAKI}:

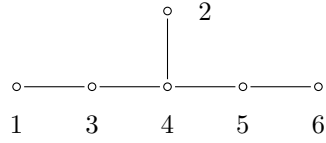
\begin{center}
\begin{tikzpicture}
\draw (7.0,-2) circle (0.05cm);
\draw(6.9,-2) -- (6.10,-2);
\draw (6.0,-2) circle (0.05cm);
\draw(7.9,-2) -- (7.10,-2);
\draw (8.0,-2) circle (0.05cm);
\draw(8.9,-2) -- (8.10,-2);
\draw (8.0,-1) circle (0.05cm);
\draw(9.9,-2) -- (9.10,-2);
\draw (9.0,-2) circle (0.05cm);
\draw(8.0,-1.9) -- (8.0,-1.1);
\draw (10.0,-2) circle (0.05cm);
\node at (6,-2.5) {$1$};
\node at (7,-2.5) {$3$};
\node at (8,-2.5) {$4$};
\node at (9,-2.5) {$5$};
\node at (10,-2.5) {$6$};
\node at (8.5,-1) {$2$};
\end{tikzpicture}
\end{center}

The Kac coordinates are $s_1=1$, $s_2=1$, $s_3=1$, $s_4=0$, $s_5=1$, $s_6=1$, see the work of Bouwknegt \cite{BOU}. The highest root has coordinates $(1,2,2,3,2,1)$ and therefore the largest eigenvalue of $\textrm{ad } r_{\textbf{s}}$ is
$$1\cdot 1 + 2\cdot 1 +2 \cdot 1 + 3 \cdot 0 + 2\cdot 1 + 1\cdot 1=8=9-1$$
and hence Condition 1 holds. As can be seen for example in \cite{GRA}, the coordinates with respect to the simple roots of the list of positive roots is
\begin{eqnarray*}
&&\boxed{(1,0,0,0,0,0)}, \boxed{(0,1,0,0,0,0)}, \boxed{(0,0,1,0,0,0)},(0,0,0,1,0,0), \boxed{(0,0,0,0,1,0)}, \boxed{(0,0,0,0,0,1)}, (1,0,1,0,0,0) \\
&& \boxed{(0,1,0,1,0,0)}, \boxed{(0,0,1,1,0,0)}, \boxed{(0,0,0,1,1,0)},(0,0,0,0,1,1), (1,0,1,1,0,0), (0,1,1,1,0,0), (0,1,0,1,1,0)\\
&& (0,0,1,1,1,0), (0,0,0,1,1,1), (1,1,1,1,0,0),(1,0,1,1,1,0), (0,1,1,1,1,0), (0,1,0,1,1,1), (0,0,1,1,1,1)\\
&&  (1,1,1,1,1,0), (1,0,1,1,1,1), (0,1,1,2,1,0), (0,1,1,1,1,1), (1,1,1,2,1,0), (1,1,1,1,1,1),(0,1,1,2,1,1)\\
&&(1,1,2,2,1,0),(1,1,1,2,1,1),(0,1,1,2,2,1),(1,1,2,2,1,1),(1,1,1,2,2,1), (1,1,2,2,2,1),(1,1,2,3,2,1)\\
&&(1,2,2,3,2,1)
\end{eqnarray*} 
We have boxed the roots of $\textbf{s}$-degree $1$.
It follows that
$$\textrm{dim } \mathfrak g_{-1}= \textrm{dim } \mathfrak g_{1}=8$$
and similarly
$$\textrm{dim } \mathfrak g_{0} = \textrm{rank }\mathfrak g + 2\cdot 1=8$$
and hence Condition 2 holds.
\end{example}
  
\begin{example} 
Consider the conjugacy class $\textrm{D}_{2n}(a_{n-1})$. For $n \le 4$ the Kac coordinates are calculated by Bouwknegt in \cite{BOU}. We now calculate the Kac coordinates for all $n$. The order of the Weyl group element is $2n$. By \cite{REE} (Proposition 2.2) all Kac coordinates are either $0$ or $1$ and $N=\sum_{i=0}^{r} a_i s_i =2n$. Since the highest root is $\alpha_1+2(\alpha_2+\cdots +\alpha_{2n-1})+\alpha_{2n-1}+\alpha_{2n}$ this implies
\begin{eqnarray}
\label{Kac-equation}
s_{0}+s_1+2(s_2+\cdots s_{2n-2}) +s_{2n-1} + s_{2n} = 2n 
\end{eqnarray}
We will show that 
\begin{eqnarray}
\label{type-D-coordinates}
\textbf{s}=(1,1,0,1,0,1,0,1,0,1 \cdots, 0,1,1)
\end{eqnarray}
In \cite{REE} a general formula for the dimension of $\mathfrak g_0$ is given. In the current situation it gives
$$\textrm{dim } \mathfrak g_{0} = \frac{hr}{2n}=\frac{(4n-2)\cdot (2n)}{2n}=4n-2$$
Our strategy is to show that unless $\textbf{s}$ is as in Equation (\ref{type-D-coordinates}), $\textrm{dim }\mathfrak g_{0}$ is bigger than $4n-2$. 

Let $t_{1},t_{2},\cdots$ be the lengths of intervals consisting of $0$'s in the sequence of Kac coordinates. So for example if $\textbf{s}= (1,1,0,0,1,1,0,1,1)$ then $t_{1}=2$ and $t_{2}=1$. The set of positive roots includes
$\sum_{i\le k <j} \alpha_{k}$ with $1\le i < j \le 2n$. Let us count how many of these are of $\textbf{s}$-degree $0$:
All $\alpha_k$'s need to be part of the same interval of $0$'s. Therefore the number of degree $0$ positive roots of this type is
$$\sum_{j} \frac{t_j (t_j+1)}{2}$$

Since the Cartan algebra is contained in $\mathfrak g_0$ and there is an equal number of positive and negative roots of $\textbf{s}$-degree $0$ it follows that
$$4n-2=\textrm{dim }\mathfrak g_{0} \ge 2n+2\cdot  \sum_{j} \frac{t_{j}(t_{j}+1)}{2} $$
Let $K_{0}=\sum_{j} t_{j}$ denote the number of Kac coordinates equal to $0$. Then
\begin{eqnarray}
\label{constraint-equation}
2n-2 \ge  K_{0}+ \sum_{j} t_{j}^{2} 
\end{eqnarray}
One sees from Equation (\ref{Kac-equation}) that there are three possibilities for $K_{0}$, namely $n-1$, $n$, and $n+1$. One also has $\sum_j t_{j}^2 \ge \sum_{j} t_j =K_0$ with equality if and only if $t_j=1$ for all $j$. Hence it follows from Equation (\ref{constraint-equation}) that $K_0=n-1$ and $t_j=1$ for all $j$.  From $K_0=n-1$ it follows from Equation (\ref{Kac-equation}) that $s_0=s_1=s_{2n-1}=s_{2n}=1$. But then the only way to avoid $t_{j}>1$ for some $j$ is to have
$$\textbf{s}=(1,1,0,1,0,1,0,1,\cdots,0,1,1)$$
as desired.

Since $s_0=1$ it follows that Condition 1 holds. We have already seen that $\textrm{dim } \mathfrak g_{0}=4n-2$, let us now calculate $\textrm{dim } \mathfrak g_{-1} = \textrm{dim }\mathfrak g_{1}$. Consider first the positive roots of the form
$\sum_{i\le k <j} \alpha_{k}$ with $1\le i < j \le 2n$. Out of these, for $\textbf{s}$-degree $1$ one obtains $n$ roots of height $1$, $2n-2$ roots of height $2$, $n-2$ roots of height $3$. There are two more types of positive roots:
\begin{enumerate}[(i)]
\item 
$\alpha_{2n}+\sum_{i\le k \le 2n-2} \alpha_{k}$ with $1\le i \le 2n-1$ (where for $i=2n-1$ the summation simply is ignored)
\item 
$\sum_{i\le k <j} \alpha_k +2 \sum_{j \le k \le 2n-2} \alpha_{k} +\alpha_{2n-1}+\alpha_{2n}$ with $1\le i < j \le 2n$
\end{enumerate}
The first case yields the two roots $\alpha_{2n}$ and $\alpha_{2n-2}+\alpha_{2n}$ of $\textbf{s}$-degree $1$ and the second case yields nothing. In total one obtains
$$\textrm{dim }\mathfrak g_{1}=n+(2n-2)+(n-2)+2=4n-2$$
and Condition 2 holds.
\end{example}

\begin{example}
In the case of $\textrm{F}_{4}(a_{1})$ the Ansatz of the previous example does not single out the Kac coordinates uniquely: The order of the Weyl group element is $6$, and the highest root is given by $2\alpha_1+3\alpha_2 +4 \alpha_3 +2\alpha_4$. Hence the Kac labels are $a_0=1$, $a_1=2$, $a_2=3$, $a_3=4$, $a_4=2$. Since all $s_i$'s are $1$ or $0$ it follows that the quintuples $(1,1,1,0,0)$ or $(1,0,1,0,1)$ or $(0,1,0,1,0)$ or $(0,0,0,1,1)$ are exactly the possible choices of $(s_0,\cdots,s_4)$ satisfying $\sum_{i=0}^{4} a_{i} s_{i}=6$.  One has
$$\textrm{dim } \mathfrak g_{0} = \frac{hr}{6}=8$$
From the list of positive roots, see for example \cite{GRA}, one directly calculates that exactly for $(1,0,1,0,1)$ or $(0,1,0,1,0)$ one has
$$\textrm{dim } \mathfrak g_{0} = 4+2\cdot 2$$
So the numerics do not distinguish between those two possibilities for the Kac coordinates and we refer to \cite{DF} (Appendix A) for details that Conditions 1 and 2 do indeed hold for $\textrm{F}_{4}(a_1)$. Assuming this, we can deduce the Kac coordinates: Since $a_{0}=1$ in order for Condition 1 to hold one needs $s_{0}=1$. Hence the Kac coordinates of $\textrm{F}_{4}(a_1)$ are $(1,0,1,0,1)$.
\end{example}

By Equation (\ref{Cartan-property}) and since $\lambda$ is semi-simple, it follows via the isomorphism $\Phi$ that
\begin{eqnarray}
\textrm{ker } \textrm{ad } \Lambda =\mathcal H_{[w]}
\end{eqnarray}
as well as
\begin{eqnarray}
\label{decomposition-equation}
\widehat{\mathfrak g}= \textrm{ker } \textrm{ad } \Lambda \oplus \textrm{Im } 
\textrm{ad } \Lambda
\end{eqnarray}
Since
$$[r_{\textbf{s}},a_1]=\frac{1}{N} \cdot a_1 \;\;\; ,\;\;\; [r_{\textbf{s}},a_{1-N}]= \frac{1-N}{N} \cdot a_{1-N}$$ 
it follows for $\Lambda=a_1 + z \cdot a_{1-N}$ that
$$\left [\partial_{z} + \frac{r_{\textbf{s}}}{z}, \Lambda \right ] = \frac{1}{Nz} \cdot \Lambda $$
More generally, it follows from our previous arguments that for the basis elements $\Lambda_j$ of $\mathcal H_{[w]}$ given in Equation (\ref{Lambda-definition}) one has
\begin{eqnarray}
\label{derivation-Heisenberg-equation}
\left[\partial_{z} + \frac{r_{\textbf{s}}}{z}, \Lambda_j  \right ]= \frac{j}{N} \cdot \Lambda_{j-N}
\end{eqnarray}

The next part of the proof follows closely the arguments by Cafasso and Wu in \cite{CAW1} (Theorem 3.11). Let $\gamma = \exp(Y)$ where $Y$ in $\widehat{\mathfrak g}$ has a decomposition with respect to the $\textbf{s}$-gradation as $Y=\sum_{i<0} Y_{i}$.  We want to show that there is such $Y$ such that 
\begin{eqnarray}
\label{degree-by-degree-equation}
\exp( \textrm{ad }Y)\left(\partial_{z}+ \frac{r_{\textbf{s}}}{z} + \Lambda \right ) =\partial_{z} +\frac{r_{\textbf{s}}}{z} +\Lambda +[Y, \partial_{z}+\frac{r_{\textbf{s}}}{z}+ \Lambda] +\frac{[Y,[Y, \partial_{z}+\frac{r_{\textbf{s}}}{z} +\Lambda] ]}{2!} + \cdots = \partial_{z} + \Lambda
\end{eqnarray}
We make the Ansatz that $Y_{j}=0$ if $j$ is not divisible by $N+1$ and show recursively with respect to $\textbf{s}$-degree that Equation (\ref{degree-by-degree-equation}) can be solved. Note that under our Ansatz one has for any indices $i_{1},\cdots,i_{k}$ that
$$\textrm{deg}_{\textbf{s}} \; [Y_{i_{1}},[Y_{i_{2}},\cdots, [Y_{i_{k}},\partial_{z} + \frac{r_{\textbf{s}}}{z}]\cdots ] \equiv - N \equiv  1 \mod N+1$$
as well as
$$\textrm{deg}_{\textbf{s}} \; [Y_{i_{1}},[Y_{i_{2}},\cdots, [Y_{i_{k}},\Lambda ]\cdots ] \equiv  1 \mod N+1$$
Hence to prove Equation (\ref{degree-by-degree-equation}) one has to check it for $\textbf{s}$-degree $-i\cdot(N+1)+1$ with $i\ge 0$. For $i=0$ the equation simply is $\Lambda=\Lambda$. The first non-trivial equation corresponds to $i=1$ and hence is in $\textbf{s}$-degree $-N$: 
\begin{eqnarray}
\label{first-equation}
\left [Y_{-N-1},\Lambda \right ]=-\frac{r_{\textbf{s}}}{z}
\end{eqnarray}
This has a solution by Equation (\ref{crucial-equation}): There is $\nu$ in $\mathfrak g_{-1}$ such that $[\nu,\Lambda]=r_{\textbf{s}}$ and hence one can take 
$$Y_{-N-1}=-\frac{\nu}{z}$$
In degree $-i\cdot (N+1)+1$ one needs to solve
\begin{eqnarray}
\label{general-equation}
[Y_{-i(N+1)},\Lambda]=-[Y_{-(i-1)(N+1)},\partial_{z}+ \frac{r_{\textbf{s}}}{z} ] -\frac{1}{2}[Y_{-(i-1)(N+1)},[Y_{-(N+1)},\Lambda]]-\frac{1}{2}[Y_{-(N+1)},[Y_{-(i-1)(N+1)},\Lambda]] + C 
\end{eqnarray}
where $C$ depends only on $Y_{-j(N+1)}$ for $j \le i-2$. We now show that there exists $Y$ solving these equations. 

By Equation (\ref{decomposition-equation}) every element $x$ of $\widehat{\mathfrak g}$ can be written uniquely as a sum $x_{1}+x_{2}$ with $x_{1}$ in the Heisenberg algebra $\mathcal H_{[w]}$ and $x_{2}$ in the image of the adjoint action of $\Lambda$. We call $x_{1}$ the Heisenberg part of $x$. Suppose now Equation (\ref{general-equation}) has been solved up to a certain $i-1$. To prove the existence of $Y_{-i(N+1)}$ with the desired properties we will show that after adding a suitable element of the Heisenberg algebra to $Y_{-(i-1)(N+1)}$ the Heisenberg part of the right-hand side of Equation (\ref{general-equation}) can be made to equal $0$. Let $H$ denote the Heisenberg part of
$$ -\frac{1}{2}[Y_{-(i-1)(N+1)},[Y_{-(N+1)},\Lambda]]-\frac{1}{2}[Y_{-(N+1)},[Y_{-(i-1)(N+1)},\Lambda]] + C $$
Then 
$$\tilde H := H\cdot \frac{N}{(i-1)(N+1)}\cdot z$$ 
is again an element of the Heisenberg algebra and of degree $-(i-1)(N+1)$. We now modify $Y_{-(i-1)(N+1)}$ by subtracting $\tilde H$. Since $\mathcal H_{[w]}$ is commutative this does not affect the validity of the equations in degree $-j(N+1)+1$ for $j \le i-2$. Using Equation (\ref{derivation-Heisenberg-equation}) one obtains 
$$\left [\partial_{z}+ \frac{r_{\textbf{s}}}{z},\tilde H \right ]=\frac{-(i-1)(N+1)}{N} \cdot \tilde H \cdot \frac{1}{z}=-H$$ 
Hence, the new Heisenberg part of $-[Y_{-(i-1)(N+1)},\partial_{z}+ \frac{r_{\textbf{s}}}{z} ]$ is obtained by subtracting $H$. Also, the Heisenberg part of $[Y_{-(i-1)(N+1)},[Y_{-(N+1)},\Lambda]]$ does not change: Since $\mathcal H_{[w]}$ is commutative, subtracting an element $\tilde H$ of the Heisenberg algebra from $Y_{-(i-1)(N+1)}$ changes $[Y_{-(i-1)(N+1)},[Y_{-(N+1)},\Lambda]]$ by 
$$-[\tilde H, [Y_{-(i-1)(N+1)},\Lambda]]=[\Lambda,[\tilde H,Y_{-(i-1)(N+1)}]]$$ 
This is in the image of $\textrm{ad } \Lambda$ and hence has $0$ Heisenberg part. It follows that the Heisenberg part of the right-hand side of Equation (\ref{general-equation}) can be made $0$. Hence there is $Y_{-i(N+1)}$ solving Equation (\ref{general-equation}) and hence by induction there exists a solution $Y$ to Equation (\ref{degree-by-degree-equation}). 

Recall that we have fixed a faithful representation $\xi : \mathfrak g \rightarrow \mathfrak g \mathfrak l(V)$ of $\mathfrak g$. In order to prove Theorem \ref{Katz-extension-theorem} it now suffices to show that, possibly after a finite pull-back, the connection 
$$\partial_{z} + \frac{r_{\textbf{s}}}{z} + \Lambda$$
can be diagonalized by an element of $\textrm{GL}_{\textrm{dim } \xi}(\mathbb{C}[z,z^{-1}])$.  Consider the weight space decomposition $V=\bigoplus_{\nu} V_{\nu}$. The eigenvalues of $\textrm{ad }r_{\textbf{s}}$ are integers after multiplication by  $N$. It follows there is $i$ with $N|i$ such that $r_{\textbf{s}}$ acts on each $V_{\nu}$ as multiplication by an element $a_{\nu}=b_{\nu}/i$ in $\ZZ/i$. For $w=z^{1/i}$ one obtains
$$ [i]^{*} \left [\partial_{z} + \frac{r_{\textbf{s}}}{z} + \Lambda \right ]\cong \partial_{w} + \frac{i\cdot r_{\textbf{s}}}{w} + i w^{i-1}(a_{1}+a_{1-N} w^{i})$$
Now consider the gauge transformation $g$ that multiplies $V_{\nu}$ by $z^{a_{\nu}}=w^{b_{\nu}}$. Then
$g \partial_{w}(g^{-1})$ acts on $V_{\nu}$ as $-b_{\nu}/w$ and hence
\begin{eqnarray}
\label{gauge-equation}
g \partial_{w}(g^{-1})=-\frac{i\cdot r_{\textbf{s}}}{w}
\end{eqnarray}
Furthermore, an element $E_{\alpha}$ in the root space of $\alpha$ takes $V_{\nu}$ to $V_{\nu + \alpha}$ and hence on each $V_{\nu}$ one has
\begin{eqnarray}
\label{second-gauge-equation}
g E_{\alpha} g^{-1} = w^{i\cdot ((\nu+\alpha)(r_{\textbf{s}} )- \nu(r_{\textbf{s}}))} E_{\alpha}=w^{i\cdot \alpha(r_{\textbf{s}})} E_{\alpha}
\end{eqnarray}
Using Equations (\ref{gauge-equation}) and (\ref{second-gauge-equation}) it follows that
\begin{eqnarray*}
\label{pull-back-equation}
g\left ( [i]^{*} \left [\partial_{z} + \frac{r_{\textbf{s}}}{z} + \Lambda \right ]  \right )  g^{-1}  & \cong & \partial_{w} -\frac{i\cdot r_{\textbf{s}}}{w}+\frac{i\cdot r_{\textbf{s}}}{w}+ i w^{i-1}(a_{1} w^{i/N} +  a_{1-N} w^{i+i(1-N)/N})\\[10pt]
&=& \partial_{w}  + i w^{\frac{i}{N}+i-1} (a_1+ a_{1-N}) 
\end{eqnarray*}
Since $N\ge i$ this connection is regular at $0$ and since $\lambda=a_{1}+a_{1-N}$ is by assumption regular semi-simple it follows that the connection can be diagonalized by a constant gauge transformation. This completes the proof of Theorem \ref{Katz-extension-theorem}.
\end{proof}

\begin{remark}
A priori, Theorem \ref{Katz-extension-theorem} does not rule out that $\partial_{z}+\Lambda$ is also its own Katz extension since it could be isomorphic to $\partial_{z}+r_{\textbf{s}}/z +\Lambda$ on the sphere. In particular, the two connections would be isomorphic on a formal disc around $0$. However, in general this does not hold as we now show. Let us consider the example of $\textrm{E}_6(a_1)$. Let $A^{-1}$ denote the inverse of the $\mathfrak e_6$ Cartan matrix. One obtains 
\begin{eqnarray*}
r_{\textbf{s}}&=&(1,1,1,0,1,1)^{\textrm{T}} A^{-1} (h_1,\cdots, h_6)^{\textrm{T}}\\[10pt]
&=&\frac{1}{9}(6h_1+8h_2+11h_3+15h_4+11h_5+ 6h_6)
\end{eqnarray*}
In one of the two faithful $27$-dimensional representations of $\mathfrak e_6$ the element $r_{\textbf{s}}$ can be conjugated to a diagonal matrix of the form
\begin{eqnarray}
\label{diagonal-matrix-equation}
r_{\textbf{s}}=\frac{1}{9} \cdot \textrm{diag}(6,5,4,4,3,3,2,2,2,1,1,1,0,0,0,-1,-1,-2,-1,-2,-2,-3,-3,-4,-4,-5,-6)
\end{eqnarray}
As described by Babbitt and Varadarajan \cite{BV} (p. 24) the semi-simple part of the monodromy of the connection $\partial_{z} + r_{\textbf{s}}/z + \Lambda$ is the exponential $\textrm{exp}(2\pi i r_{\textbf{s}})$. Since in Equation (\ref{diagonal-matrix-equation}) not all entries of $r_{\textbf{s}}$ are integers, it follows that the monodromy is non-trivial. However, the monodromy of $\partial_{z} + \Lambda$ is trivial around $0$ and it follows that $\partial_{z}+\Lambda$ is not its own Katz extension.
\end{remark}

\hspace{0.2in}

\textbf{Acknowledgements:} It is a great pleasure to thank Mattia Cafasso and Chao-Zhong Wu for very helpful exchanges. Thanks also to important corrections and remarks by the referee.

\end{document}